\documentclass[conference]{IEEEtran}
\IEEEoverridecommandlockouts
\usepackage{amsmath}
\usepackage{amssymb,amsfonts}
\usepackage{amsthm}
\usepackage{cite}
\usepackage{cases}
\usepackage{url}
\usepackage{indentfirst}
\usepackage{accents} 
\usepackage{subfig}
\usepackage{algpseudocode}
\usepackage{graphicx}
\usepackage{textcomp}
\usepackage{caption}
\usepackage{commath}
\usepackage{xcolor}
\usepackage{cleveref}
\usepackage{mathtools}
\usepackage[letterpaper, left=0.625in, right=0.625in, bottom=1in, top=0.7in]{geometry}
\theoremstyle{remark}

\newtheorem{theorem}{\bf \emph{Theorem}}
\newtheorem{lemma}{\bf \emph{Lemma}}

\newtheorem{remark}{\textit{Remark}}

 \def\cF{{\mathcal{F}}}  
   \def\cL{{\mathcal{L}}}

\def\ba{{\mathbf{a}}} \def\bb{{\mathbf{b}}}  \def\bd{{\mathbf{d}}} \def\be{{\mathbf{e}}} 
\def\bff{{\mathbf{f}}} \def\bg{{\mathbf{g}}}   
   \def\bn{{\mathbf{n}}} 
   \def\bs{{\mathbf{s}}} \def\bt{{\mathbf{t}}}
   \def\bx{{\mathbf{x}}} \def\by{{\mathbf{y}}}

\def\bA{{\mathbf{A}}} \def\bB{{\mathbf{B}}} \def\bC{{\mathbf{C}}} \def\bD{{\mathbf{D}}} 
\def\bF{{\mathbf{F}}} \def\bG{{\mathbf{G}}} \def\bH{{\mathbf{H}}} \def\bI{{\mathbf{I}}}

  \def\bW{{\mathbf{W}}} \def\bX{{\mathbf{X}}}

  \def\R{{\mathbb{R}}} \def\C{{\mathbb{C}}}

\def\blockdiag{\mathop{\mathrm{blk}}}

\def\argmin{\mathop{\mathrm{argmin}}}

\def\mod{\mathop{\mathrm{mod}}}

     \def\d4{\!\!\!\!}

         \def\ls{\left\{}

  \def\-{\! - \!}  \def\+{\! + \!}  \def\={\! = \!}  \def\>{\! > \!}

\newcommand{\bef}{\begin{figure}}
\newcommand{\eef}{\end{figure}}
\newcommand{\beq}{\begin{eqnarray}}
\newcommand{\eeq}{\end{eqnarray}}

\def\BibTeX{{\rm B\kern-.05em{\sc i\kern-.025em b}\kern-.08em
		T\kern-.1667em\lower.7ex\hbox{E}\kern-.125emX}}
\usepackage{titlesec}
\titlespacing*{\section}
{0pt}{2.3ex}{1.3ex}
\titlespacing*{\subsection}
{0pt}{1.3ex}{0.8ex}

\begin{document}
\setlength{\textfloatsep}{0.5\baselineskip plus 0.2\baselineskip minus 0.5\baselineskip}
	\title{Joint Delay and Phase Precoding Under True-Time Delay {Constraints} for THz Massive MIMO}
	\author{\IEEEauthorblockN{Dang Qua Nguyen and Taejoon Kim}
	\IEEEauthorblockA{\textit{Department of Electrical Engineering and Computer Science} \\
		\textit{University of Kansas, Lawrence, KS 66045 USA}\\
		Email: \{quand, taejoonkim\}@ku.edu}
		\thanks{This work was supported in part by the National Science Foundation (NSF) under Grants CNS1955561, and in part by the Office of Naval Research (ONR) under Grant N00014-21-1-2472.}
}
\maketitle	

	\begin{abstract}
    A new approach is presented to the problem of compensating the beam squint effect arising in wideband terahertz (THz) hybrid massive multiple-input multiple-output (MIMO) systems, based on the joint optimization of the phase shifter (PS) and true-time delay (TTD) values under per-TTD device time delay constraints.  
    Unlike the prior approaches, the new approach does not require the unbounded time delay assumption; the range of time delay values that a TTD device can produce is strictly limited in our approach. 
    Instead of focusing on the design of TTD values, we jointly optimize both the TTD and PS values to effectively cope with the practical time delay constraints.
    Simulation results that illustrate the performance benefits of the new method for the beam squint compensation are presented.  Through simulations and analysis, we show that our approach is a generalization of the prior TTD-based precoding approaches.  
    \end{abstract}
	\section{Introduction}
	\label{secI}
        Communications in the terahertz (THz) band (0.1-10 THz) have recently attracted significant interests from academia and industry due to the availability of the tens or hundreds of gigahertz bandwidth\cite{Han2014}.
        In THz communications, data rates on the orders of 10 to 100 Gbps are achievable using the currently available digital modulation techniques. 
        To deal with the high path losses, power consumption, and inter-symbol interference found in the THz communication channels, the combination of hybrid massive multiple-input multiple-output (MIMO) and orthogonal frequency division multiplexing (OFDM) technologies has been popularly discussed recently. 

        It is well known that as the number of antennas grows, the system can be greatly simplified in terms of beamforming and precoding complexity \cite{Marzetta2010}.  
        This has engendered a significant interest in massive MIMO systems at sub-$6$ GHz \cite{Marzetta2010,Rusek2013, Larsson14} and  millimeter-wave (mmWave) frequencies \cite{ayach2014,Alkhateeb2014,Hadi2016}.
        The underlying assumption behind several strong results \cite{Marzetta2010,Rusek2013,Larsson14,ayach2014,Alkhateeb2014,Hadi2016} was narrowband. 
        Unlike the narrowband systems, wideband THz OFDM systems suffer from the array gain loss across different subcarriers as the number of antennas grows due to beam squint \cite{Han2021,Cai2016,Wang2019}.  
       Beam squint refers to a phenomenon in which deviation occurs in the spatial direction of each OFDM subcarrier when wideband OFDM is used in a very large antenna array system. 
        This causes a sizable array gain loss, which potentially demotivates the use of OFDM in THz massive MIMO. 
        

        \subsection{Related Works}
        To deal with the beam squint effect, beam broadening techniques have been previously studied in mmWave massive MIMO systems \cite{Cai2016,Liu2019}.
        While they show efficacy in the mmWave bands, these techniques cannot be directly extended to THz due to the extremely narrow pencil beam requirements imposed by one or two orders of magnitude higher carrier frequencies. 
        Traditionally, the beam squint effect has been independently studied in the radar community (e.g., \cite{Mailloux2017, Longbrake2012}, and references therein).
        These traditional works \cite{Mailloux2017, Longbrake2012} proposed to employ true-time delay (TTD) lines to relieve the beam squint effect.
        Recently, the TTD methods have been proposed for THz hybrid massive MIMO-OFDM systems by introducing delay-phase subarray architectures \cite{tan2019,Gao2021,Matthaiou2021}.
        Most of these prior works have focused on (i) the design of TTD values with fixed  phase shifter (PS) values and (ii) the assumption that the TTD values could increase linearly with the number of antennas without bounds.
        However, it is practically difficult to implement these approaches. First, the unbounded TTD assumption cannot be realized in practice due to the inevitable limitation of a TTD device. The range of time delay values that a TTD device can produce is strictly limited (e.g., $\leq 508$ ps \cite{Cho2018}). Second, to cope with the practical TTD constraints, it is much desirable to jointly optimize both the TTD and PS values to compensate the beam squint. 
         
        \subsection{Overview of Methodology and Contributions}
        In this paper, we present an approach to the problem of beam squint compensation in wideband THz hybrid massive MIMO-OFDM systems, based on the joint optimization of the PS and TTD values subject to per-TTD device time delay constraints. 
        First, assuming a  fully-digital array, we characterize the optimal unconstrained analog precoder that completely compensates the beam squint effect.   
        Then, the problem is formulated as a joint PS and TTD optimization problem that minimizes the distance between the {optimal unconstrained analog precoder} and the product of PS and TTD precoders.
         Although the formulated problem is non-convex and thus difficult to solve directly, we show that by transforming the problem into the phase domain, the original problem is converted to an equivalent convex problem, which allows us to find a closed-form of the global optimal solution.
         Our analysis reveals the amount of time delay and the number of transmit antennas required given an acceptable level of beam squint compensation. 
        
    \textit{Notation:} A bold lower case letter $\bx$ is a column vector and a bold upper case letter $\bX$ is a matrix. 
    $\bX^T$, $\bX^H$, $\|\bX\|_F$, $X(i,j)$, and $|x|$  are, respectively, the transpose, conjugate transpose, Frobenius norm, $i$th row and $j$th column entry of $\bX$, and the modulus of $x\in \C$. 
    $\blockdiag[\bx_1, \bx_2, \dots, \bx_N]$ is an $nN \times N$ block diagonal matrix such that its main-diagonal blocks contain $\bx_i \in \R^n$, for $i =1,2,\dots,N$ and all off-diagonal blocks are zero.
    $\cF_{m,n}$ denotes the set of all $\bX \in \C^{m \times n}$ such that $|X(i,j)| = \frac{1}{\sqrt{m}}, \forall  i, j$. 
    $\mathbf{0}_n$, $\mathbf{1}_{n}$, and $\bI_{n}$ denote, respectively, the $n\times 1$ all-zero vector, $n\times 1$ all-one vector, and $n\times n$ identity matrix.
    Given $\bx \in \R^n$, $e^{j\bx}$ denotes $[e^{jx_1}, e^{jx_2}, \dots, e^{jx_n}]^T \in \C^n$.
    \section{Channel Model and Motivation}
    In this section, we describe the channel model and the beam squint effect of large antenna array systems.
    \subsection{Channel Model}
    We consider the downlink of a THz hybrid massive MIMO-OFDM system where the base station is equipped with an $N_t$-element uniform linear array (ULA) with element spacing $d$. 
    The ULA is fed by $N_{RF}$ radio frequency (RF) chains to simultaneously transmit $N_s$ data streams to an $N_r$-antenna user.
    The dimensions of $N_t$, $N_r$, $N_s$, and $N_{RF}$ satisfy $N_r = N_s \leq N_{RF} \ll N_t$. 
    Herein, $f_c$ and $B$ denote the central (carrier) frequency and bandwidth of the OFDM system, respectively. 
    We let $K$ be the number of subcarriers and an odd number. Then, the $k$th subcarrier frequency is given by $f_k = f_c + \frac{B}{K}(k-1-\frac{K-1}{2}).$
    The frequency domain MIMO-OFDM channel at the $k$th subcarrier is represented by 
    \begin{equation}
    \label{eq1}
    \bH_k = \sqrt{\frac{N_tN_r}{L}} \sum_{l=1}^{L}\alpha_l e^{-j 2 \pi \tau_l f_k}\bff(N_t,\psi_{k,l})\bff^H(N_r,\phi_{k,l}), 
    \end{equation}                  
    where $L$ denotes the number of spatial paths, and $\psi_{k,l} = 2d\frac{f_k}{\nu}\sin({\Psi_l})$ and $\phi_{k,l} = 2d\frac{f_k}{\nu}\sin({\varphi_l})$ are the spatial directions of the $k$th subcarrier on the $l$th path at the transmitter and receiver, respectively, where $\nu = 3\times 10^8$ m/s is the speed of light, and $\Psi_{l} \in [-\frac{\pi}{2}, \frac{\pi}{2}]$ and $\varphi_{l} \in [-\frac{\pi}{2}, \frac{\pi}{2}]$ are the angel-of-departure (AoD) and angel-of-arrival (AoA) of the $l$th path, respectively. 
    The $\alpha_l \in \C$ and $\tau_l \in \R$ in \eqref{eq1} represent the gain and delay of the $l$th path, respectively, and $\bff(N,\psi) =\frac{1}{\sqrt{N}} [1,e^{-j\pi\psi},\dots,e^{-j\pi(N-1)\psi}]^T$ is the array response vector of an $N$-element ULA at the direction $\psi$. 

    Assuming $d = \frac{\nu}{2 f_c}$, the spatial directions at the central frequency $f_c$ are simplified to $\psi_{c,l} = \sin(\Psi_{l})$ and $\phi_{c,l} = \sin(\varphi_l)$, $\forall l$. 
    Hence, setting $\zeta_k = \frac{f_k}{f_c}$ leads to $\psi_{k,l} = \zeta_k \psi_{c,l}$ and $\phi_{k,l} = \zeta_k \phi_{c,l}$, $\forall k,l$. 
    For ease of exposition, we assume that $L=N_{RF}$, which is equivalent to setting $\alpha_l =0$, $\forall l > L$ in (1) whenever $L<N_{RF}$. 

    \subsection{Motivation}
    \label{beamsquint}
     As aforementioned in \emph{Section}~\ref{secI}, when wideband OFDM is employed in a massive MIMO system, a substantial array gain loss at each subcarrier could occur due to beam squint. To describe the beam squint effect, we consider the $l$th path of the channel in \eqref{eq1}. We denote the beam forming vector matched to the array response vector with the AoD $\Psi_l$ as $\bff_l = \bff(N_t,\psi_{c,l})$. 
    The array gain at the $k$th subcarrier of the $l$th path is then given by $    g(\bff_l,\psi_{k,l}) = |\bff^H(N_t,\psi_{k,l})\bff_l|$, i.e.,    
    \begin{equation}
    \label{eq2}
         g(\bff_l,\psi_{k,l}) \!=\! \frac{1}{N_t}\bigg|\!\sum\limits_{n=0}^{N_t-1}\!\!e^{jn\pi(\psi_{k,l}-\psi_{c,l})} \bigg| \!=\! \left|\frac{\sin \left( N_t\Delta_{k,l}\right)}{N_t\sin \left(\Delta_{k,l}\right)}\!\right|\!,
    \end{equation}
    where $\Delta_{k,l} = \frac{\pi}{2}(\psi_{k,l}-\psi_{c,l})$. 
     It is not difficult to observe that at the central frequency, the  array gain is $g(\bff_l,\psi_{c,l}) = 1$, because $\lim_{x \rightarrow 0} \frac{\sin (N_tx)}{\sin (x)} = N_t$. 
    However, when $f_k \neq f_c$,  $\Delta_{k,l}$ deviates from zero. 
    As a result, any non-central subcarrier suffers from the array gain loss. 
    The implication of beam squint in the spatial domain is that the beams at non-central subcarriers may completely squint.
    
    The beam squint effect becomes severe either when the bandwidth B increases or when $N_t$ grows. The following lemma quantifies the asymptotic array gain loss as $N_t \rightarrow \infty$.      
    \begin{lemma}
    \label{lm1} 
    Suppose that $\psi_{k,l}$ is the spatial direction at the $k$th subcarrier ($f_k \neq f_c$) of the $l$th path. 
    Then, the array gain in \eqref{eq2} converges to $0$ as $N_t$ tends to infinity, i.e., $\lim_{N_t \rightarrow \infty} g(\bff_l, \psi_{k,l}) = 0.$
    \end{lemma}
    \begin{proof}
    The  array gain in \eqref{eq2} can be rewritten as
      $ g(\bff_l,\psi_{k,l}) = \frac{1}{N_t} \left|\frac{\sin (N_t\Delta_{k,l})}{\pi\Delta_{k,l}} \frac{\pi \Delta_{k,l}}{\sin (\Delta_{k,l})}\right|,$ 
    where $\Delta_{k,l} \neq 0$ because $f_k \neq f_c$. 
    Then the lemma follows from the definition of the Dirac delta function $\lim_{N_t \rightarrow \infty}\frac{\sin (N_t\Delta_{k,l})}{\pi\Delta_{k,l}} = \delta(\Delta_{k,l})$ \cite{bracewell2000}, which completes the proof.
    \end{proof}
     \begin{figure}[ht]
        \centering
        \includegraphics[width = 0.4\textwidth]{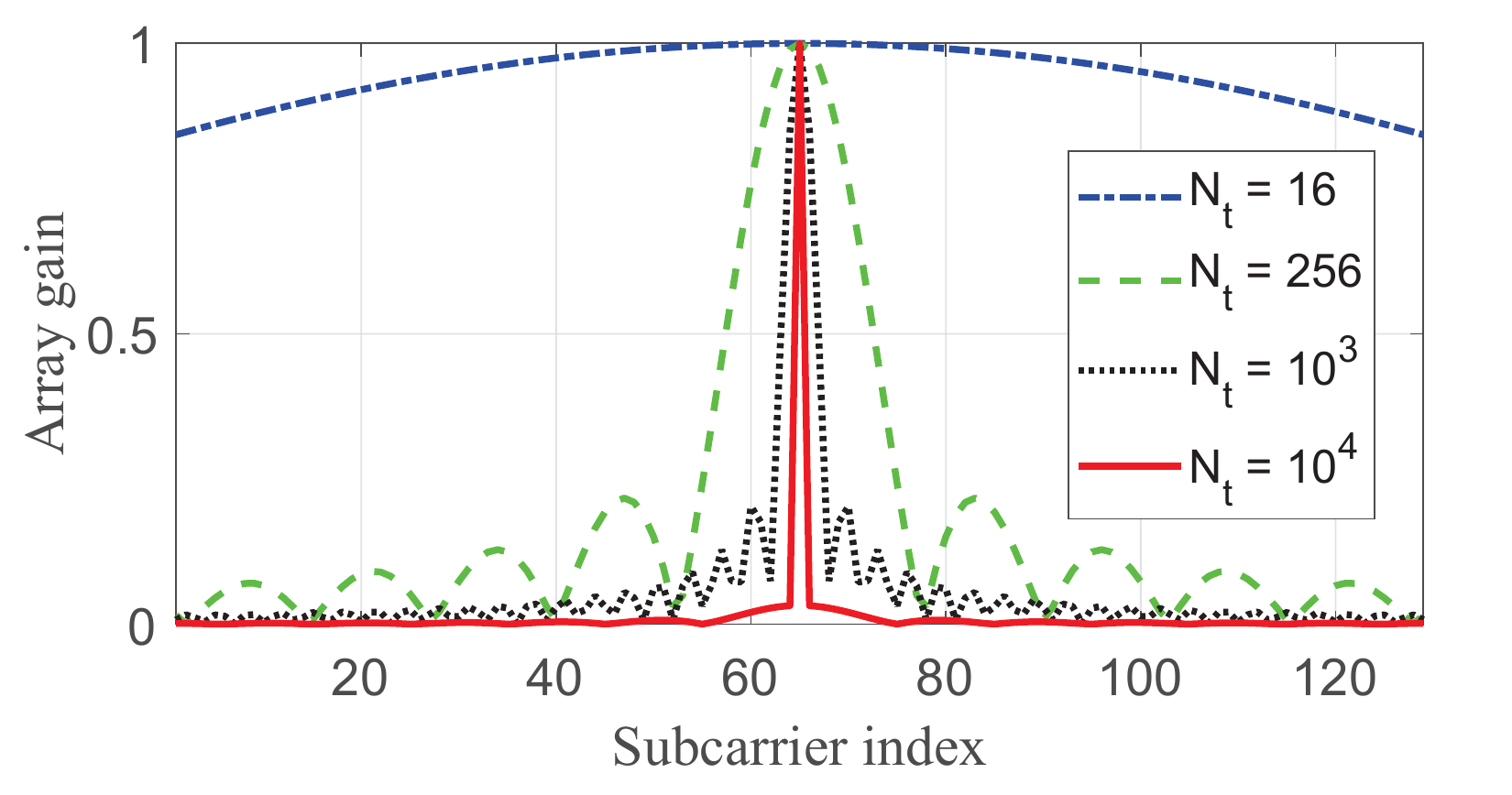}
        \caption{Array gain vs. subcarriers for different numbers of transmit antennas ($N_t$).}
        \label{fig1}
    \end{figure}
    Fig. \ref{fig1} illustrates the convergence trend of \emph{Lemma}~1. The array gain patterns are calculated for $f_c = 300$ GHz, $K = 129$, $\psi_{c,l}= 0.8$, and $B= 30$ GHz. 
    As $N_t$ tends to be large, the maximum array gain is only achieved at the central frequency, i.e., $65$th subcarrier in Fig.\ref{fig1}, while other subcarriers suffer from nearly 100\% array gain losses. 
    This is quite opposite to the traditional narrowband massive MIMO system in which the array gain grows as $N_t \rightarrow \infty$.
    
    Wideband THz massive MIMO research is in its early stages. In order to truly unleash the potential of THz communications, a hybrid precoding mechanism that can effectively compensate for the beam squint effect under practical constraints is of paramount importance.
    \section{Hybrid Precoding Under TTD Contraint}    
    We consider a TTD-based hybrid precoding architecture \cite{tan2019,Gao2021}, where each RF chain drives $M$ TTDs and each TTD is connected to $N = \frac{N_t}{M}$ phase shifters (PSs) as shown in Fig. \ref{fig2}. 
    \begin{figure}[t]
        \centering
        \includegraphics[width=0.5\textwidth]{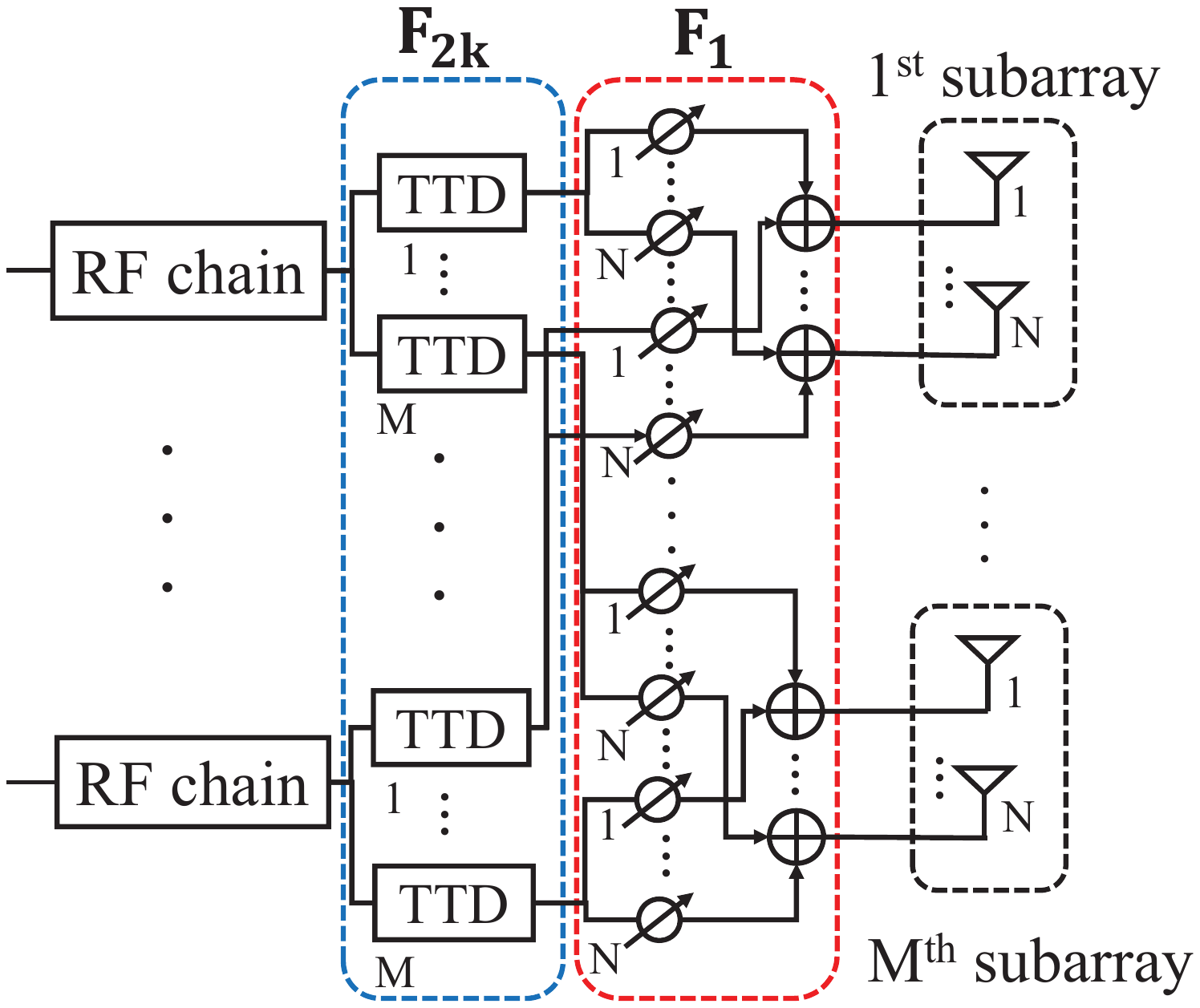}
        \caption{TTD-based hybrid precoding architecture}
        \label{fig2}
    \end{figure}
   The $N_t$-element ULA is divided into $M$ subarrays with $N$ antennas per subarray ($N_t = MN$).      
    A TTD delays the $k$th subcarrier by $t$ ($0\leq t \leq t_{\max}$), which corresponds to the $-2\pi f_k t$ phase rotation. 
    The $t_{\max}$ is the maximum time delay value that a TTD device can produce.    
     The received signal $\by_k \in \C^{N_r}$ at the $k$th subcarrier is then expressed by 
     \begin{equation}
         \label{eq:signalmodel}
         \by_k = \bH^{H}_k\bF_1\bF_{2,k}(\{\bt_l\}^{N_{RF}}_{l=1})\bW_k\bs_k + \bn_k,
     \end{equation}
    where $\bs_k \in \C^{N_s}$, $\bW_k \in \C^{N_{RF} \times N_s}$, and $\bn_k \in \C^{N_r}$ are, respectively, the transmitted data stream, baseband digital precoder, and noise.
    Herein, $\bF_1$=$ \frac{1}{\sqrt{N_t}}$$[\bG_1,\bG_2, \dots,  \bG_{N_{RF}}]$$\in \C^{N_t \times MN_{RF}}$ is the PS precoding matrix, where $\bG_l = \blockdiag[e^{j \pi \bx^{(l)}_1},e^{j \pi \bx^{(l)}_2},\dots,e^{j \pi \bx^{(l)}_{M}}] \in \C^{N_t \times M}$ is the PS submatrix and  $\bx^{(l)}_m = [x^{(l)}_{1,m},x^{(l)}_{2,m},\dots,x^{(l)}_{N,m}]^T\in \R^N$ is the PS vector with the PSs connected to the $m$th TTD on the $l$th RF chain, $\forall l,m$. 
    The $\bF_{2,k}(\{\bt_l\}^{N_{RF}}_{l=1})\! =\! \blockdiag[e^{-j2\pi f_k \bt_1},e^{-j2\pi f_k \bt_2}, \dots, e^{-j2\pi f_k \bt_{N_{RF}}}] \in \C^{MN_{RF} \times N_{RF}}$ is the time delay precoding matrix, where $\bt_l = [t^{(l)}_1,t^{(l)}_2,\dots, t^{(l)}_{M}]^T \in \R^{M}$ is the time delay vector on the $l$th RF chain, $\forall l$.   
   
  \section{Preliminaries}
   In this section, we describe the sign invariance property of the array gain and identify the {optimal unconstrained analog precoder} that will be incorporated in \emph{Section}~\ref{sec4}.   
  \subsection{Sign Invariance of Array Gain}
   Denoting the $l$th column of $\bF_1\bF_{2,k}(\{\bt_l\}^{N_{RF}}_{l=1})$ in \eqref{eq:signalmodel} as $\bg^{(l)}_{k} = \frac{1}{\sqrt{N_t}}\bG_le^{-2\pi f_k \bt_l}$, the array gain associated with $\bg^{(l)}_{k}$ is given, based on \eqref{eq2}, by     
    \begin{equation}
    \label{eq3}
     g(\bg^{(l)}_{k},\psi_{k,l}) = \frac{1}{N_t}\bigg|
       \!\sum\limits_{m=1}^{M}\!\sum\limits_{n=1}^{N}\! e^{j \pi \zeta_{k} \gamma^{(l)}_{n,m}} e^{j \pi  x^{(l)}_{n,m}} e^{-j\pi \zeta_k  \vartheta^{(l)}_{m}}\!\bigg|, \hspace{-0.2cm}
    \end{equation} 
    where 
        $\gamma^{(l)}_{n,m} = ((m-1)N+n-1)\psi_{c,l}$ and $\vartheta^{(l)}_{m} = 2f_ct^{(l)}_m \in [0,\vartheta_{\max}]$ with $\vartheta_{\max} = 2f_c t_{\max}$. The $g(\bg^{(l)}_{k},\psi_{k,l})$ in \eqref{eq3} is invariant to the multiplication of negative signs to $\gamma_{n,m}^{(l)}$, $x_{n,m}^{(l)}$, and $\vartheta^{(l)}_{m}$. 
        To be specific, given $\psi_{c,l} > 0$ (i.e., $\gamma^{(l)}_{n,m} > 0$, $\forall m$, $n$), we denote $\{ x_{m,n}^{(l) \star} \}$ and $\{\vartheta^{(l)\star}_{m}\}$ as the optimal values that maximize $g(\bg^{(l)}_{k},\psi_{k,l})$. 
        Then, it is not difficult to observe that $\{-x^{(l) \star}_{n,m}\}$ and $\{\vartheta_{\max} - \vartheta^{(l)\star}_{m}\}$ also maximize $g(\bg^{(l)}_{k},-\psi_{k,l})$. 
        This is a straightforward property when we simultaneously change the sign of the phase in every term inside the double sums in \eqref{eq3}.
        Leveraging this sign invariant property, in what follows, we assume that  $\psi_{c,l} \geq 0$, $\forall l$.
        This sign invariant property will be exploited when deriving the solution to our optimization problem in \emph{Section}~\ref{sec4} and \emph{Appendix}~\ref{SecondAppendix}.

  \subsection{{Optimal Unconstrained Analog Precoder}}
  
  We identify an {optimal unconstrained analog precoder} that maximizes the array gain of each subcarrier by completely compensating the beam squint effect. 
  The purpose is to provide a reference design.
  To this end, we define $\widetilde{\bF}_k \in \cF_{N_t,{N_{RF}}}$ as an unconstrained analog precoder where $\bff^{(l)}_k \in \cF_{N_t,1}$ denotes the $l$th column of $\widetilde{\bF}_k$. 
   The array gain obtained by $\bff^{(l)}_k$ is given by  $g(\bff^{(l)}_k, \psi_{k,l}) = \big|\bff^H(N_t,\psi_{k,l})\bff^{(l)}_{k}\big|$.
   By the Cauchy-Schwarz inequality,
    $g(\bff^{(l)}_k, \psi_{k,l} ) \leq \left\|\bff(N_t,\psi_{k,l})\right\|_{2} \|\bff^{(l)}_k\|_{2} = 1$, 
   where the equality holds if and only if $\bff^{(l)}_{k} = \bff(N_t,\psi_{k,l})$, resulting in $\widetilde{\bF}^{\star}_k = [\bff(N_t,\psi_{k,1}),\bff(N_t,\psi_{k,2}), \dots,\bff(N_t,\psi_{k,N_{RF}})]$, where the $((m-1)N+n-1)$th row and $l$th column entry of $\widetilde{\bF}^{\star}_k$ is 
    \begin{equation}
        \label{eq4}
        \widetilde{F}^{\star}_k((m-1)N+n,l) =\frac{1}{\sqrt{N_t}} e^{-j\pi\zeta_k\gamma^{(l)}_{n,m}}, \forall k, l,m, n.
    \end{equation} 
    It is noteworthy to point out that the entries in \eqref{eq4} is only realizable when $N_{RF} = N_t$, i.e., each antenna is fed by its dedicated RF chain, which is highly impractical for THz massive MIMO systems. 
    Therefore, we attempt to best approximate the {optimal unconstrained analog precoder} in \eqref{eq4}  as a product of $\bF_1$ and $\bF_{2,k}(\{\bt_{l}\}_{l=1}^{N_{RF}})$.
   \section{Joint Phase and Delay Precoding under TTD {Constraints}}
    \label{sec4}
    Ideally, we wish to find $\{\bt_l\}^{N_{RF}}_{l=1}$ and $\bF_1$  satisfying $\bF_1\bF_{2,k}(\{\bt_l\}^{N_{RF}}_{l=1}) = \widetilde{\bF}^{\star}_k$, $\forall k$, i.e., $e^{j \pi x^{(l)}_{n,m}- j\pi\zeta_k \vartheta^{(l)}_{m}} = e^{-j \pi 
    \zeta_k \gamma^{(l)}_{n,m}}$, $\forall k, l, m,n$. However, given fixed $l$, $m$, and $n$, solving $K$-coupled matrix equations is an ill-posed problem, because PSs only generate fixed phase values (frequency-independent), while TTDs generate frequency-dependent phase values. To overcome this issue, we approach to formulate a problem that optimizes $\bF_1$ and $\{\bt_l\}^{N_{RF}}_{l=1}$ by minimizing the difference between $\bF_1\bF_{2,k}(\{\bt_l\}^{N_{RF}}_{l=1})$ and $\widetilde{\bF}^{\star}_{k}$, $\forall k$:
    \vspace{-0.3cm}
   \begin{subequations}
        \label{opt_prob}
        \beq
\d4\d4        \min_{\bF_1, \{\bt_l\}_{l=1}^{N_{RF}}} \d4&&\d4 \frac{1}{K}\sum_{k=1}^{K} \left\|\widetilde{\bF}^{\star}_k - \bF_1\bF_{2,k}(\{\bt_l\}^{N_{RF}}_{l=1})\right\|^2_F, \label{obj1}\\       
\d4\d4 \text{subject to} \d4&&\d4 0\leq t^{(l)}_{m} \leq t_{max}, \forall l,m,\label{ct1}\\
\d4\d4        \d4&&\d4 |\bF_1(i,j)| = \frac{1}{\sqrt{N_t}},~\text{ if } |\bF_1(i,j)| \neq 0, \forall i,j, \label{ct3}\\ 
\d4\d4        \d4&&\d4 |\bF_{2,k}(p,q)| = 1,~\text{ if } |\bF_{2,k}(p,q)| \neq 0,\forall p,q, \label{ct4}\\ 
\d4\d4        \d4&&\d4 \bF_1\bF_{2,k}(\{\bt_l\}^{N_{RF}}_{l=1}) \in \cF_{N_t, N_{RF}}, \forall k.\label{ct5}
        \eeq
   \end{subequations}
     The constraint \eqref{ct1} indicates the range of time delay values of the per-TTD device. The constraint \eqref{ct5} describes the constant modulus property of the product $\bF_1\bF_{2,k}(\{\bt_l\}^{N_{RF}}_{l=1})$ at every subcarrier. We note here that while we are focusing on precoding in this paper, a similar optimization problem can be induced to design receive combiners to deal with the beam squint effect at the receiver.     
     
      The constraints in \eqref{ct3}, \eqref{ct4}, and \eqref{ct5} and the coupling between $\bF_1$ and $\bF_{2,k}(\{\bt_{l}\}^{N_{RF}}_{l=1})$ in \eqref{obj1} make the problem difficult to solve. 
       Besides, \eqref{opt_prob} can be viewed as a matrix factorization problem under non-convex constraints, which has been also studied in wireless communications in the context of hybrid analog-digital precoding \cite{ayach2014,Hadi2016,Zhang2018,JunZhang2016,Sohrabi2017}. 
       A common approach was applying block coordinate descent (BCD) and relaxing the constraints \cite{ayach2014,Hadi2016,Zhang2018,JunZhang2016,Sohrabi2017} to deal with the non-convexity. 
       Unlike the prior approaches, we show in this work that the original non-convex problem in \eqref{opt_prob} can be readily converted into an equivalent convex problem.
     \begin{lemma}
     \label{lm2}
      For $x_0 \in \R$ and $y \in \R$, the following equality holds
        $\argmin\limits_{y:\mod(y,\pi) \neq x_0} |e^{jx_0} -e^{jy}|$ $= \argmin\limits_{y:\mod(y,\pi) \neq x_0}|x_0-y|$, where $\mod(y,\pi)$ is $y$ modulo $\pi$.  
    \end{lemma}
    \begin{proof}
    See Appendix \ref{FirstAppendix}. 
    \end{proof}
 
    Rewriting the objective function in \eqref{opt_prob} gives
    \begin{equation}
    \label{eq6}
     \frac{1}{N_t}\frac{1}{K}\sum_{k=1}^{K}\!\sum_{l=1}^{N_{RF}} \!\sum_{m = 1}^{M}\!\sum_{n=1}^{N} \bigg|e^{-j\pi \zeta_k \gamma^{(l)}_{n,m}} - e^{j \pi x^{(l)}_{n,m}} e^{-j\pi \zeta_k \vartheta^{(l)}_{m} }\bigg|^2\!\!\!, \!\!\!\!    
    \end{equation}
    and incorporating \emph{Lemma} \ref{lm2} into \eqref{eq6} converts \eqref{opt_prob} into the following equivalent problem:
     \begin{subequations}
        \label{opt_pro_2}
        \beq
       \d4\d4\!\! \min_{\{x^{(l)}_{n,m}\}, \{\vartheta^{(l)}_{m}\}}\d4&& \frac{1}{K} \!\sum\limits_{k=1}^{K}\!\!\sum\limits_{l=1}^{N_{RF}} \!\!\sum\limits_{m=1}^{M}\!\sum\limits_{n=1}^{N}\left|
          x^{(l)}_{n,m} \!\!-\! \zeta_k\vartheta^{(l)}_{m}\!\! +\! \zeta_k\gamma^{(l)}_{n,m}\right|^2\!\!\!\!, \label{obj2}\\       
       \d4\d4\d4  \text{subject to} \d4&& 0 \leq \vartheta^{(l)}_{m}  \leq \vartheta_{\max}, ~\forall  l, m.\label{ct6}
        \eeq
    \end{subequations}
    
     \noindent The PS and TTD variables in \eqref{opt_pro_2} are readily collected into a composite matrix $\bA_l = [\ba^{(l)}_{1},\dots,\ba^{(l)}_{M}] \in \R^{(N+1)\times M}$, where $\ba^{(l)}_{m} = [\bx^{(l) T}_m,\vartheta^{(l)}_{m}]^T \in \R^{N+1}$. Containing the unconstrained analog counterpart in \eqref{eq4} in a matrix  $\bB^{(l)}_k \in \R^{N \times M},$ where $B^{(l)}_k(n,m) = -\zeta_k \gamma^{(l)}_{n,m}$, $\forall k,l,n,m$, the problem \eqref{opt_pro_2} becomes 
    \begin{subequations}
    \label{opt_pro3}
    \beq
        \underset{\{\bA_l \}^{N_{RF}}_{l=1}}\min &&\frac{1}{K}\sum\limits_{k=1}^{K}\sum\limits_{l=1}^{N_{RF}} \left\|\bC_k\bA_{l} - \bB^{(l)}_k\right\|^2_F, \\
        \text{subject to}  
        &&\mathbf{0}^T_{M} \preceq \be^{T}_{N+1}\bA_{l} \preceq \vartheta_{\max}\mathbf{1}^T_{M}, ~\forall l, \label{eq:9b} 
        \eeq 
    \end{subequations}
    where $\bC_k = \begin{bmatrix}
     \bI_{N},-\zeta_k \mathbf{1}_{N}
     \end{bmatrix} \in \R^{N \times (N+1)}$ and $\be_{N+1} = [\mathbf{0}^T_N,1]^T \in \R^{N+1}$. The $\preceq$ is an entry-wise vector inequality. 
  
    By introducing $\bC \= \frac{1}{K}\sum_{k=1}^{K}\bC^T_k\bC_k \in \R^{(N+1)\times(N+1)}$ and $\bD_l \= \frac{1}{K}\sum_{k=1}^{K}\bC^T_k\bB^{(l)}_k \!\in\! \R^{(N+1)\times M}$, the problem \eqref{opt_pro3} is equivalently 
    \vspace{-0.2cm}
    \begin{subequations}
        \label{eq:lcqp}
    \beq
           \underset{\ba^{(l)}_m}\min ~ && \ba^{(l)T}_m\bC\ba^{(l)}_m -2\bd^{(l) T}_m\ba^{(l)}_m,\label{eq:10a}\\
        \text{subject to } 
        && 0 \leq \be_{N+1}^{T}\ba^{(l)}_m \leq \vartheta_{\max}, \forall l,m, 
    \eeq
    \end{subequations} where $\bd^{(l)}_m$ in \eqref{eq:10a} is the $m$th column of $\bD_l$.
    The problem \eqref{eq:lcqp} is solved globally and the solution to it is summarized below.
    \begin{theorem}
    \label{theorem1}
    The optimal solution $\ba^{(l)\star}_m = [\bx^{(l) \star T}_m,\vartheta^{(l)\star}_m]^T$ to \eqref{eq:lcqp} is given by 
    \begin{equation}
    \normalfont
    \label{eq:theorem1}
    x^{(l) \star}_{n,m} = 
    \ls \!\!\! \begin{array}{ll}
    \frac{N-2n+1}{2}\psi_{c,l},& \d4 \text{ if } 0 \leq \psi_{c,l} \leq \frac{4f_c t_{\max}}{(2m-1)N-1}, \\ 
     \vartheta_{\max} - \gamma^{(l)}_{n,m},&  \d4 \text{ otherwise}, \forall l,n,m,
    \end{array} \right.
    \end{equation}
    and $\vartheta^{(l)\star}_m = 2f_c t^{(l)\star}_m$, where the $t^{(l)\star}_m$ is    
    \begin{equation}
        \normalfont
    \label{eq:TTDvalues}
    t^{(l) \star}_m = 
    \ls \!\!\! \begin{array}{ll}
    \frac{(2m-1)N-1}{4f_c}\psi_{c,l}, \d4 &\text{ if } 0 \leq \psi_{c,l} \leq \frac{4f_c t_{\max}}{(2m-1)N-1},  \\  t_{\max}, \d4 &  \text{ otherwise},\forall l,m.
    \end{array} \right.
    \end{equation}
    \end{theorem}
    
    \begin{proof}
     See Appendix \ref{SecondAppendix}. 
    \end{proof}
    \begin{remark}
    \label{rmk2}
    The PS and TTD solutions in \eqref{eq:theorem1} and \eqref{eq:TTDvalues} are different from those in prior works. For instance, the PS values in \cite{tan2019} were not optimized and given by $x^{(l)}_{n,m} = -(n-1)\psi_{c,l}$, $\forall n$, which is distinguishable from \eqref{eq:theorem1}. The time delay value of the $m$th TTD was $t^{(l)}_{m} = m\frac{N\psi_{c,l}}{2f_c}$ in \cite{tan2019}; as $m$ increases, the $t^{(l)}_m$ could be greater than $t_{\max}$, in which case such $t^{(l)}_m$ needs to be knocked down to the $t_{\max}$, resulting in performance deterioration. On the other hand, as will be discussed in \emph{Section}~\ref{sec6} in detail, when all TTD values are smaller than $t_{\max}$, the approaches in \cite{tan2019,Gao2021} achieve the same array gain performance as the proposed approach, meaning that the designs in \cite{tan2019,Gao2021} is a special case of \emph{Theorem}~\ref{theorem1}. 
    \end{remark}
    
    Focusing on the large $t_{\max}$, the TTD value in \eqref{eq:TTDvalues} becomes $t^{(l) \star}_m = \frac{((2m-1)N-1)\psi_{c,l}}{4f_c},$ which is proportional to the delay increase at every subcarrier and thereby, best compensating the beam squint. On the basis of this observation, we obtain selection criteria (rule of thumb) on the required number of transmit antennas and value of maximum time delay based on \eqref{eq:TTDvalues} as follows.

     \emph{$N_t$ Selection Criterion:} Given $t_{\max}$ value determined by the employed TTD devices, choose $N_t$ such that 
     \begin{equation}
     \label{eq:criterion1}
         N_t \leq \frac{M}{2m-1} + \frac{4Mf_c t_{\max}}{(2m-1)\psi_{c,l}}, \forall l, m.
     \end{equation}
    
    \emph{$t_{\max}$ Selection Criterion:} Equivalently, given $N_t$ value determined by the transmission array, the TTD device should be chosen to satisfy 
    \begin{equation}
    \label{eq:criterion2}
        t_{\max} \geq \psi_{c,l}\frac{(2m-1)N_t-M}{4Mf_c}, \forall l, m.
    \end{equation}
    These bounds are once again approximations, but they give insight into the values of $N_t$ and $t_{\max}$ required. For example, these bounds depend on a random variable $\psi_{c,l}$, which can be characterized by the angle spread statistics $\psi_{c,l} \in [\psi_{\min},\psi_{\max}]$ of the deployed array antenna systems, leading to sufficient conditions for \eqref{eq:criterion1} and \eqref{eq:criterion2} by replacing $\psi_{c,l}$ with $\psi_{\max}$. 
    \section{Simulations}
    \label{sec6}
    In this section, we conduct numerical simulations to evaluate the average array gain performance of the proposed joint PS and TTD precoding and verify our analysis in \emph{Section}~\ref{sec4}. 
    \textcolor{red}{}
    The average array gain of the $l$th path is computed as $\frac{1}{K}\sum_{k=1}^K g(\bg^{(l)}_k,\psi_{c,l})$. For the numerical simulation, we consider $B = 30$ GHz, $f_c = 300$ GHz, $K = 129$, and $\psi_{c,l} =0.8$, and the number of TTDs to be $M = 16$.       
    \label{sectionIV}
     \begin{figure}[h]
        \centering
        \includegraphics[width = 0.4\textwidth]{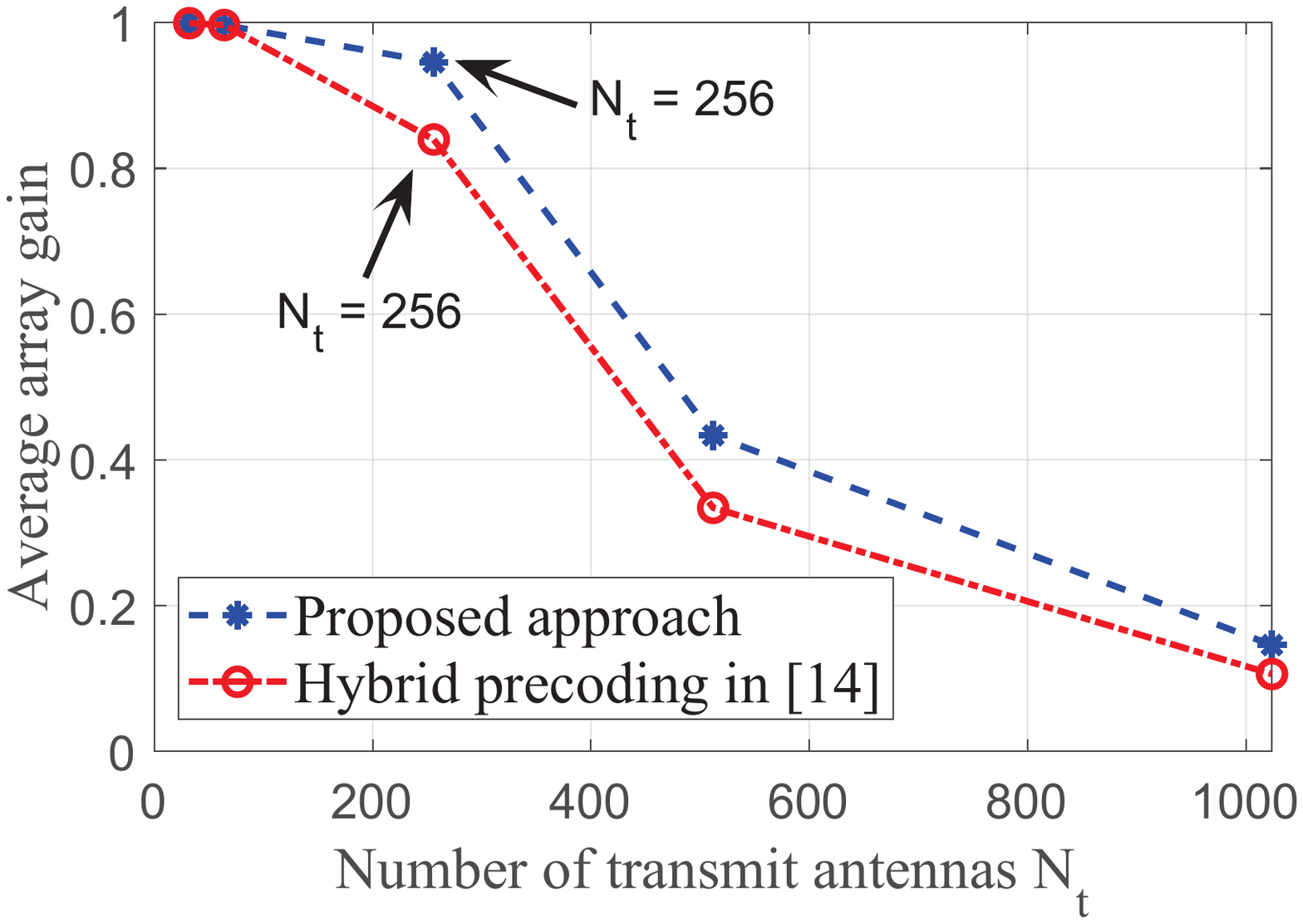}
        \caption{The average array gain vs. $N_t$ for $t_{\max} = 340$ ps.}
        \label{fig7}
    \end{figure}
    
     Fig.~\ref{fig7} displays the average array gain curves of the proposed approach and the prior approach in \cite{tan2019} for the number of antennas $N_t \in \{32,64,256,512,1024\}$ and $t_{\max} = 340$ ps. When $N_t = 32$ or $64$, both methods obtain the same near-optimal array gain performance, because the designed time delay values in both approaches are smaller than $340$ ps. However, when $N_t$ increases further, e.g., $N_t \geq 256$, both approaches suffer from array gain loss as shown in Fig. 3. Nevertheless, the proposed approach reveals a rather reliable performance than the benchmark because some TTD values produced by the prior approach are larger than the $t_{\max}$ and this increases the array gain loss. 
     \begin{figure}[h]
        \centering
        \includegraphics[width = 0.4\textwidth]{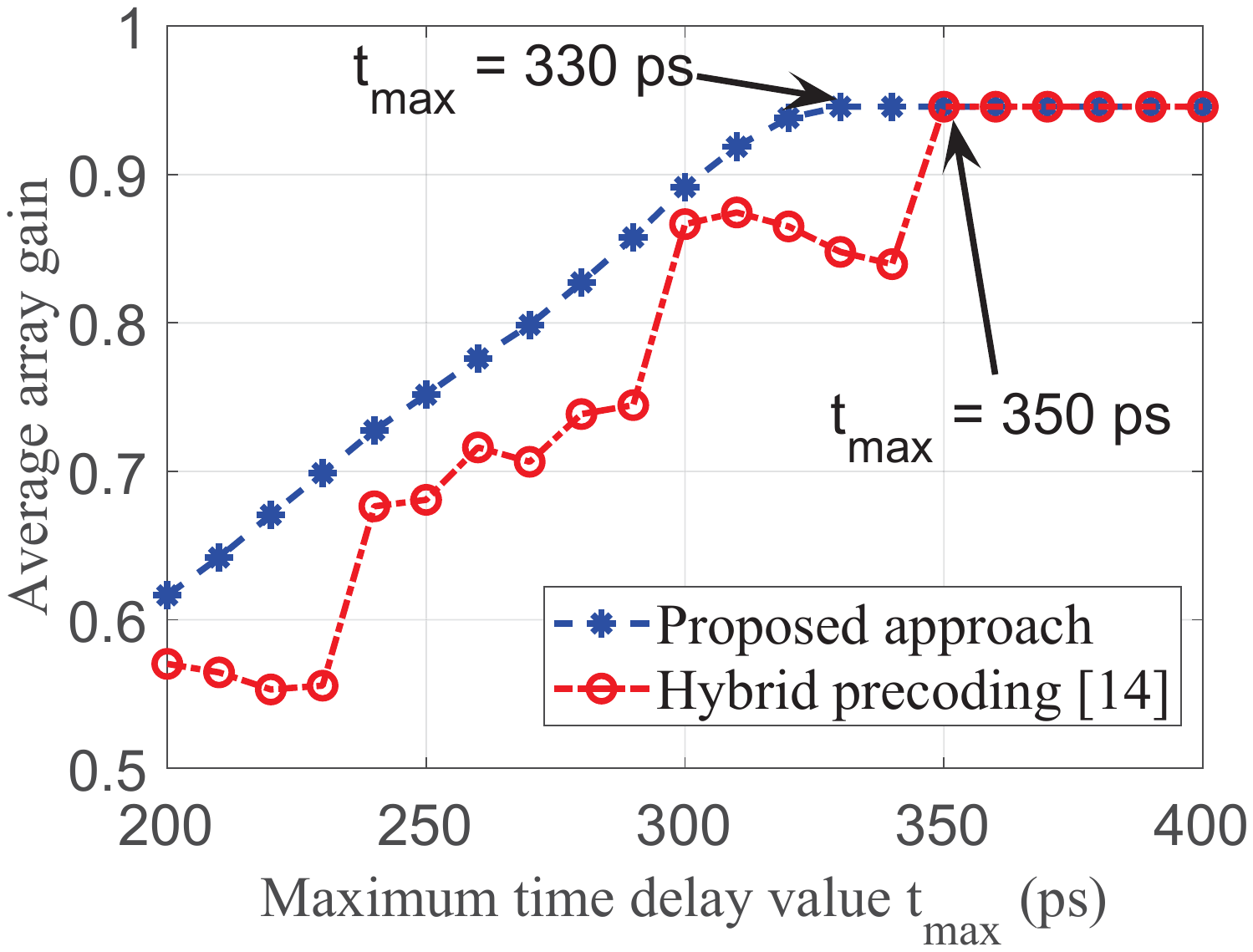}
        \caption{The average array gain vs. $t_{\max}$ for $N_t = 256$.}
    \label{fig6}  
    \end{figure}
    
    In Fig.~\ref{fig6}, the average array gain performance of the proposed approach is compared with the prior approach in \cite{tan2019}, in which we increase the $t_{\max}$  values from $200$ ps to $400$ ps while fixing $N_t=256$. 
     When $t_{\max} \in [200, 340]$ ps, as seen from the figure, the proposed approach reveals consistent performance increase while the benchmark \cite{tan2019} shows inconsistency because of the infeasible TTD values designed. When $ t_{\max} \geq 330$ ps, the proposed approach steadily converges to the average array gain performance around 94\%. Meanwhile, the benchmark achieves the same performance when $t_{\max} \geq 350$ ps. 
     
     Overall, the trends in Figs.~\ref{fig7}-\ref{fig6} reaffirm \emph{Remark}~\ref{rmk2} stating that the prior design \cite{tan2019} is a special case of the proposed approach. Figs.~\ref{fig7}-\ref{fig6} also reveal the intuition behind the selection criteria in \eqref{eq:criterion1} and \eqref{eq:criterion2}, in which $N_t \leq 256$ and $t_{\max} \geq 330$ ps are chosen to meet the criteria, respectively.  
    
    \section{Conclusion}
    A TTD-based hybrid precoding approach was proposed to compensate the beam squint effect in the  wideband THz massive MIMO-OFDM system by jointly optimizing the PS and TTD precoders under the per-TTD device time delay {constraints}. 
    The joint optimization problem was initially formulated in the context of minimizing the distance between the {optimal unconstrained analog precoder} and the product of the PS and TTD precoders.
   By transforming the original problem to a tractable phase domain, we identified the global optimal solution. 
   Based on the closed-form expression of our solution, we proposed the selection criteria for the required number of transmit antennas and value of maximum time delay. 
   Through analysis and simulations, we affirmed that our proposed design is a generalization of the prior TTD precoding approaches. 
    \appendices
  \section{Proof of lemma 2}
  \label{FirstAppendix}
      We assume that $0<|x_0  \!-\!  y|<\pi$ without loss of generality. 
      Then,  
        $\argmin_{0 <|x_0-y| <\pi}|e^{jx_0}\-e^{jy}|$ $=$ $\underset{0 <|x_0-y|<\pi}{\argmin}\Big|\sin\left(\frac{x_0-y}{2}\right)\Big|$ $=$ $\underset{0 <|x_0-y| <\pi}{\argmin}\sin \Big(\Big|\frac{x_0-y}{2}\Big| \Big)$ $= \underset{0 <|x_0-y| <\pi}{\argmin}$ $|x_0-y|$, 
    where the last step follows from the fact that $\sin(\cdot)$  is a strictly increasing function in $(0,\frac{\pi}{2})$. This completes the proof. 
    \section{Proof of theorem 1}
    \label{SecondAppendix}
     To show \emph{Theorem}~\ref{theorem1}, we need to first provide closed-form expressions of $\be^T_{N+1}\!\bC^{-1}\!\be_{N+1}$ and $\be_{N+1}^T\!\bC^{-1}\bd^{(l)}_{m}$, which will be later used  in this proof. To this end, we first claim that    
    \begin{subequations}
     \label{eq:KKT1}
    \beq
        \be^T_{N+1}\bC^{-1}\be_{N+1} \d4&=&\d4 \frac{1}{\eta}, \label{eq:16a} \\  \be_{N+1}^T\bC^{-1}\bd^{(l)}_{m} \d4&=&\d4 \frac{(2m-1)N-1}{2}\psi_{c,l} \label{eq:16b},
    \eeq
    \end{subequations}
    where $\eta$$=$$ \frac{NB^2}{f^2_c}\frac{(K^2-1)}{12K^2}.$
    We first prove \eqref{eq:16a}. Given $\bC_k$ in \eqref{opt_pro3}, we have 
    $\bC^T_k\bC_k \!=\! \begin{bsmallmatrix} 
     \bI_{N}\! &\! -\zeta_k \mathbf{1}_{N}  \\-\zeta_k \mathbf{1}^T_{N} \!& \!N\zeta^2_k
    \end{bsmallmatrix}$, and . After some algebraic manipulations, it is readily verified that
     $\bC =
        \begin{bsmallmatrix}
        \bI_{N} & -\mathbf{1}_{N} \\
        -\mathbf{1}^T_{N} & \Gamma \end{bsmallmatrix}$ and $\bC^{-1} =
        \begin{bsmallmatrix}
        (\bI_{N} + \frac{1}{\eta}\mathbf{1}_{N}\mathbf{1}^T_N) & \frac{1}{\eta}\mathbf{1}_{N} \\
        \frac{1}{\eta}\mathbf{1}^T_{N} & \frac{1}{\eta} \end{bsmallmatrix}$, where $\Gamma = N+\eta$. Now, it is straightforward to conclude $\be_{N+1}^T\bC^{-1}\be_{N+1} = \frac{1}{\eta}$ and $\be_{N+1}^T\bC^{-1}\bd^{(l)}_{m} = \frac{\mathbf{1}^T_{N+1}\bd^{(l)}_{m}}{\eta}.$
    From the definition of $\bd^{(l)}_m$ in \eqref{eq:lcqp}, we obtain $\bd^{(l)}_{m} \= \frac{1}{K}\sum_{k=1}^{K} \bC^T_k\bb^{(l)}_{k,m}$, where $\bb^{(l)}_{k,m} \= [-\zeta_k\gamma^{(l)}_{n,m},\dots,-\zeta_k\gamma^{(l)}_{N,m}]^T$ is the $m$th column of $\bB^{(l)}_k$ in \eqref{opt_pro3}. Since $\bd^{(l)}_{m} \= \frac{1}{K}\sum_{k=1}^{K}[\bI_{N},-\zeta_k\mathbf{1}^T_{N}]^T \bb^{(l)}_{k,m} \= \frac{1}{K}\sum_{k=1}^{K}[\bb^{(l)}_{k,m}, -\zeta_k\mathbf{1}^T_{N}\bb^{(l)}_{k,m}]^T,$
    we have $\be_{N+1}^T\bC^{-1}\bd^{(l)}_{m} \= \frac{\mathbf{1}^T_{N+1}\bd^{(l)}_{m}}{\eta} \=  \frac{(2m-1)N-1}{2}\psi_{c,l},$ completing the proof of \eqref{eq:16b}. 
    
    Now, we are ready to prove \emph{Theorem}~\ref{theorem1}. We start by formulating the Lagrangian of  \eqref{eq:lcqp},
        \begin{multline}
        \label{eq:Lagrangian}
        \cL(\ba^{(l)}_m,\lambda_1, \lambda_2) =  \ba^{(l)T}_m\bC\ba^{(l)}_m -2\bd^{(l) T}_m\ba^{(l)}_m + \\  \lambda_1(\be_{N+1}^T\ba^{(l)}_m -  \vartheta_{\max}) +  \lambda_2(-\be_{N+1}^T\ba^{(l)}_m),
    \end{multline}
     where $\lambda_1 \geq 0$ and $\lambda_2 \geq 0$ are the Largrangian multipliers.
    After incorporating the first order necessary condition for $\ba_m^{(l)}$ in \eqref{eq:Lagrangian},  the KKT conditions of \eqref{eq:lcqp} are given by: 
    \begin{equation}
         \left\{ \begin{array}{ll}
        \label{KKT}
        2\bC \ba^{(l)}_m -2\bd^{(l)}_m + \lambda_1 \be_{N+1} - \lambda_2\be_{N+1} & = 0,  \\
         \lambda_1(\be_{N+1}^T\ba^{(l)}_m\! - \! \vartheta_{\max}) &= 0, \\ 
         \lambda_2(-\be_{N+1}^T\ba^{(l)}_m) &= 0,\\
        \lambda_1 \geq 0, \lambda_2 \geq 0. 
        \end{array} \right.
    \end{equation}  
    It is obvious that when $\lambda_1 > 0$ and $\lambda_2 >0$, there is no solution.   
    When $\lambda_1 = \lambda_2 = 0$, solving \eqref{KKT} yields $\ba^{(l)}_{m} = \bC^{-1}\bd^{(l)}_m$ for $0 \leq \be_{N+1}^T\bC^{-1}\bd^{(l)}_m \leq \vartheta_{\max}$. 
    When $\lambda_1 > 0$ and $\lambda_2 = 0$, \eqref{KKT} yields $\ba^{(l)}_m = \bC^{-1}(\bd^{(l)}_{m} - \frac{1}{2}\lambda_1 \be_{N+1})$ for $ \be_{N+1}^T\bC^{-1}\bb^{(l)}_{m} > \vartheta_{\max}$. 
    However, when $\lambda_1 = 0$ and $\lambda_2 >0$, we obtain $\ba^{(l)}_m = \bC^{-1}(\bd^{(l)}_m + \frac{1}{2} \lambda_2\be_{N+1})$ for $\be_{N+1}^{T}\bC^{-1}\bd^{(l)}_m < 0$, which contradicts with \eqref{eq:KKT1}, because from the sign invariance property of array gain, we have $\frac{(2m-1)N-1}{2}\psi_{c,l} \geq 0$, $\forall \psi_{c,l} \geq 0.$ In short, solving \eqref{KKT} leads to
    \begin{equation}
    \label{eq:closed_form_KKT}
    \normalfont
    \ba^{(l) \star}_{m} \!\!=\!\! 
    \ls \!\!\! \begin{array}{ll}
    \bC^{-1}\bd^{(l)}_m, \d4\d4& \text{ if } 0 \leq \be^T_{N+1}\bC^{-1}\bd^{(l)}_m \leq \vartheta_{\max}, \\ 
    \bC^{-1}(\bd^{(l)}_{m}\! -\frac{1}{2}\! \lambda_1 \be_{N+1}),\d4\!\!\!&  \text{ if } \be^T_{N+1}\bC^{-1}\bd^{(l)}_m > \vartheta_{\max},
    \end{array} \right.
    \end{equation}
    where $\lambda_1 = \frac{\be^{T}_{N+1}\bC^{-1}\bd^{(l)}_m - \vartheta_{\max}}{\be^{T}_{N+1}\bC^{-1}\be_{N+1}}$.    
    Substituting \eqref{eq:KKT1} into \eqref{eq:closed_form_KKT} results in the closed-form solution of \eqref{eq:lcqp} as in \eqref{eq:theorem1} and \eqref{eq:TTDvalues}, which concludes the proof.  
\bibliographystyle{IEEEtran} 
\bibliography{biblib}    

\end{document}